\documentclass[a4paper]{article}
\usepackage[dvips]{graphics}
\usepackage{amsmath}

\newtheorem{theorem}{Theorem}

\newtheorem{definition}{Definition}

\newtheorem{example}{Example}

\newenvironment{proof}[1][Proof]{\noindent\textbf{#1.} }{\ \rule{0.5em}{0.5em}}

\begin{document}

\title{\textbf{Codes over Hurwitz integers}}
\author{Murat G\"{u}zeltepe  \\
{\small Department of Mathematics, Sakarya University, TR54187 Sakarya,
Turkey}}
\date{}
\maketitle

\begin{abstract}
In this study, we obtain new classes of linear codes over Hurwitz
integers equipped with a new metric. We refer to the metric as
Hurwitz metric. The codes with respect to Hurwitz metric use in
coded modulation schemes based on quadrature amplitude modulation
(QAM)-type constellations, for which neither Hamming metric nor
Lee metric. Also, we define decoding algorithms for these codes
when up to two coordinates of a transmitted code vector are
effected by error of arbitrary Hurwitz weight.
\end{abstract}


\bigskip \textsl{AMS Classification:}{\small \ 94B05, 94B15, 94B35, 94B60}

\textsl{Keywords:\ }{\small Block codes, Cyclic codes, Syndrome
decoding}

\section{Introduction }
Hamming and Lee distances have been revealed to be inappropriate
metrics to deal with quadrature amplitude modulation (QAM) signal
sets and other related constellations. To solve this problem,
different authors have constructed new error-correcting codes over
fields or rings. For example, Huber discovered a new way to
construct codes for two-dimensional signals in terms of Gaussian
integers, i.e., the integral points on the complex plane
\cite{Huber}. His original idea is to regard a finite field as a
residue field of the Gaussian integer ring modulo a Gaussian prime
and, by Euclidean division, to get a unique element of minimal
norm in each residue class, which represents each element of
finite field. Therefore, each element of finite field can be
represented by a Gaussian integer with the minimal Galois norm in
the residue class; and the set of the selected Gaussian integers
is called a constellation. Since the Galois norm of integral
points on the complex plane coincides with the Euclidean metric,
Huber's constellation is of minimal energy. Moreover, Huber
introduced the Mannheim weight by means of the Manhattan metric of
the constellation, and obtained linear codes which are of one
Mannheim error-correcting capability. In \cite{Huber2}, Huber
developed his wonderful idea further to the Eisenstein integers,
i.e., the algebraic integers of the cyclotomic field generated by
the sixth roots of unity. Although Huber's work constitutes a
relevant contribution, unfortunately the Mannheim distance is not
a true metric as was proved in\cite{Carmen2}. Later, T. P. da
Nobrega Neto \emph{et al}. in \cite{Neto} discussed the algebraic
integer rings of quadratic fields which are Euclidean norm, and
proposed a new class of linear codes. In \cite{Neto}, codes over
the ring $\mathcal{Z}[i]$ of Gaussian integers and codes over the
ring $A_p[\rho]$ of Eisenstein-Jacobi integers were presented. The
metric used in \cite{Neto} is inspired by Mannheim metric.

On the other hand, C. Martinez \emph{et al}. introduced a metric
called Lipschitz metric in \cite{Carmen2}  and obtained codes over
Lipschitz integers with respect to this metric.

In this paper, we introduce Hurwitz metric over Hurwitz integers
and give codes over Hurwitz integers with respect to this metric.
Also, we give decoding algorithms of these codes.

In what follows, we consider the following:

\begin{definition} \cite{Ramo} The Hamilton Quaternion Algebra over the set of the real numbers
($R $), denoted by $H(R)$, is the associative unital algebra given
by the following representation:

i)$H(\mathcal{R})$ is the free $\mathcal{R}$ module over the
symbols $1,\widehat{e}_1,\widehat{e}_2,\widehat{e}_3$, that is, $
H(\mathcal{R}) = \{ a_0  + a_1 \widehat{e}_1 +$ $ a_2
\widehat{e}_2 + a_3 \widehat{e}_3:\;a_0 ,a_1 ,a_2 ,a_3  \in
\mathcal{R}\}$;

ii)1 is the multiplicative unit;

iii) $ \widehat{e}_1^2  = \widehat{e}_2^2  = \widehat{e}_3^2  =  -
1$;

iv) $ \widehat{e}_1\widehat{e}_2 =  - \widehat{e}_2\widehat{e}_1 =
\widehat{e}_3,\;\widehat{e}_3\widehat{e}_1 =  -
\widehat{e}_1\widehat{e}_3 =
\widehat{e}_2,\;\widehat{e}_2\widehat{e}_3 = -
\widehat{e}_3\widehat{e}_2 = \widehat{e}_1$.
\end{definition}

The set of Lipschitz integers $H(\mathcal{Z})$, which is defined
by $H(\mathcal{Z}) = \left\{ {{a_0} + {a_1}{\widehat{e}_1} + }
\right.$ $\left. {{a_2}{\widehat{e}_2} +
{a_3}{\widehat{e}_3}:{a_0},{a_1},{a_2},{a_3} \in \mathcal{Z}}
\right\}$, is a subset of $H(\mathcal{R})$, where $\mathcal{Z}$ is
the set of all integers. If $ q = a_0  + a_1 \widehat{e}_1 + a_2
\widehat{e}_2 + a_3 \widehat{e}_3$ is a quaternion integer, its
conjugate quaternion is $ q^* = a_0 - (a_1 \widehat{e}_1 + a_2
\widehat{e}_2 + a_3 \widehat{e}_3)$. The norm of  $q$ is $ N(q) =
q q^* = a_0^2  + a_1^2 + a_2^2 + a_3^2$. The units of
${H(\mathcal{Z})}$ are $ \pm 1,\pm \widehat{e}_1, \pm
\widehat{e}_2,\pm \widehat{e}_3$.

\begin{definition}\cite{Carmen2} Let $ \pi$ be an odd integer quaternion. If
there exists $\delta \in{H(\mathcal{Z})}$ such that
$q_1-q_2=\delta \pi$ then $q_1,q_2 \in {H(\mathcal{Z})}$ are right
congruent modulo $\pi$ and it is denoted as $q_1 \equiv_rq_2$.

\end{definition}

This equivalence relation is well-defined. Hence, it can be
considered as the quotient ring of the quaternion integers modulo
this equivalence relation, which is denoted by
$$H(\mathcal{Z})_\pi=\left\{ {\left. {q\;(\bmod \pi )} \right|\;q
\in H(\mathcal{Z})} \right\} .$$ This set coincides with the
quotient ring of the integer quaternions over the left ideal
generated by $\pi$, which is denoted by $\left\langle \pi
\right\rangle $ \cite{Carmen2}.

\begin{definition}\cite{Carmen2} Let $ \pi \ne 0 $ be a quaternion integer. Given $\alpha, \beta \in
H(\mathcal{Z})_\pi$, Lipschitz distance between $\alpha$ and
$\beta$ is computed as $\left| {{a_0}} \right| + \left| {{a_1}}
\right| + \left| {{a_2}} \right| + \left| {{a_3}}
 \right|$ and is denoted by $d_\pi(\alpha,\beta)$, where $$\alpha  - \beta { \equiv _r}{a_0} + {a_1}\widehat{e}_1 + {a_2}\widehat{e}_2 +
 {a_3}\widehat{e}_3 \ (mod\ \pi)$$
 with $\left| {{a_0}} \right| + \left| {{a_1}}
\right| + \left| {{a_2}} \right| + \left| {{a_3}}
 \right|$ minimum.
\end{definition} Lipschitz weight of the element $\gamma$ is defined as  $\left| {{a_0}} \right| + \left| {{a_1}}
\right| + \left| {{a_2}} \right| + \left| {{a_3}}
 \right|$ and is denoted by $w_{L}(\gamma)$, where
 $\gamma=\alpha-\beta$ with $\left| {{a_0}} \right| + \left| {{a_1}}
\right| + \left| {{a_2}} \right| + \left| {{a_3}}
 \right|$ minimum.

More information which are related with the arithmetic properties
of $H(\mathcal{Z})$ can be found in \cite{Carmen, Carmen2, Ramo}.

\begin{theorem} \cite{Carmen2} Let $\pi \in H(\mathcal{Z})$. Then $H(\mathcal{Z})_\pi$ has
$N(\pi)^2$ elements.
\end{theorem}

\begin{definition}\cite{Con} The set of all Hurwitz integers is
$$\begin{array}{c}
 \mathcal {H} = \left\{ {{a_0} + {a_1}{\widehat{e}_1} + {a_2}{\widehat{e}_2} + {a_3}{\widehat{e}_3} \in H( R):{a_0},{a_1},{a_2},{a_3} \in \mathcal{Z}\;{\rm{or}}\;{a_0},{a_1},{a_2},{a_3} \in \mathcal{Z} + \frac{1}{2}} \right\} \\
  = H\left( \mathcal{Z} \right) \cup H\left( {\mathcal{Z} + \frac{1}{2}} \right). \\
 \end{array}$$

\end{definition}
It can be checked that $\mathcal{H}$ is closed under quaternion
multiplication and addition, so that it forms a subring of the
ring of all quaternions. \

\begin{definition} We define the set $\mathcal{R}$ as  $$\mathcal{R} =\left\{ {a + bw:\ a,b \in
\mathcal{Z}} \right\}.$$ Here and thereafter, $w$ will denote
$\frac{1}{2}(1 + \widehat{e}_1 + \widehat{e}_2 + \widehat{e}_3)$.
Let $ \pi $ be a prime in $\mathcal{R}$. If there exists $\delta
\in{\mathcal{R}}$ such that $q_1-q_2=\delta \pi$ then $q_1,q_2 \in
{\mathcal{R}}$ are congruent modulo $\pi$. We will denote it as
$q_1 \equiv q_2 \ (\bmod \ \pi)$.

\end{definition}

This equivalence relation is well-defined. We can consider the
subring of the Hurwitz integers modulo this equivalence relation,
which we denote as
$$\mathcal{R}_\pi=\left\{ {\left. {q\;(\bmod \pi )} \right|\;q
\in \mathcal{R}} \right\}.$$ It is obvious that $\mathcal{R}_\pi$
is a finite field with cardinal number $N(\pi)$.

For example, let
$\pi=1+2\widehat{e}_1+2\widehat{e}_2+2\widehat{e}_3=-1+4w$, then
$${R_\pi } = \left\{
\begin{array}{l}
 0,1, - 1 - w, - w,1 - w,2 - w, - 1 + 2w, \\
 1 - 2w, - 2 + w, - 1 + w,w,1 + w, - 1 \\
 \end{array} \right\}.$$

\begin{definition} Let $ \pi$ be a prime in $H(\mathcal{Z})$. If
there exists $\delta \in H(\mathcal{Z})$ such that $q_1-q_2=\delta
\pi$ then $q_1,q_2 \in \mathcal{H}$ are right congruent modulo
$\pi$ and it is denoted as $q_1 \equiv_rq_2$.

\end{definition}

We will use right congruent modulo $\pi$ in the present paper
unless told otherwise. Analogous results hold for left congruent
modulo $\pi$.

\begin{theorem} Let $\alpha$ be a prime integer quaternion. Then $\mathcal{H}_\alpha$ has
$2N(\alpha)^2-1$ elements.
\end{theorem}

\begin{proof} Let $\pi$ be a prime integer quaternion. According to Theorem 1, the cardinal number of
$ H(\mathcal{Z})_\pi$ is equal to $N(\pi)^2$. Also, the cardinal
number of $ H(\mathcal{Z}+\frac{1}{2})_\pi$ is equal to
$N(\pi)^2$. $( H(\mathcal{Z})_\pi -\left\{ {0} \right\}) \cap
(H(\mathcal{Z}+\frac{1}{2})_\pi-\left\{ {0} \right\})=\emptyset$
since the elements of the set  $
H(\mathcal{Z}+\frac{1}{2})_\pi-\left\{ {0} \right\}$ are defined
in the form $q-\delta \pi =
a_0+a_1\widehat{e}_1+a_2\widehat{e}_2+a_3\widehat{e}_3+a_4w$,
where $q\in H(\mathcal{Z}+\frac{1}{2}),$ $\delta,\;\pi \in
H(\mathcal{Z})$, $\ a_0,a_1,a_2,a_3\in \mathcal{Z}$ and $a_4$ is
an odd integer. But the additive identity is an element of both
sets $ H(\mathcal{Z})_\pi$ and $ H(\mathcal{Z}+\frac{1}{2})_\pi$.
Hence the proof is completed.

\end{proof}

Note that if $\delta$ is chosen from $\mathcal{H}$ instead of
$H(\mathcal{Z})$ then, Theorem 2 does not hold.

\ In the following definition, we introduce Hurwitz metric.

\begin{definition}Let $ \pi$ be a prime quaternion integer. Given $\alpha={a_0} + {a_1}\widehat{e}_1 + {a_2}\widehat{e}_2 +
 {a_3}\widehat{e}_3+a_4w , \beta={b_0} + {b_1}\widehat{e}_1 + {b_2}\widehat{e}_2 +
 {b_3}\widehat{e}_3+b_4w \in
\mathcal{H}_\pi$, then the distance between $\alpha$ and $\beta$
is computed as $\left| {{c_0}} \right| + \left| {{c_1}} \right| +
\left| {{c_2}} \right| + \left| {{c_3}}\right|+ \left| {{c_4}}
\right|$ and denoted by $d_H(\alpha,\beta)$, where $$\gamma
=\alpha - \beta { \equiv _r}{c_0} + {c_1}\widehat{e}_1 +
{c_2}\widehat{e}_2 +
 {c_3}\widehat{e}_3+c_4w \ (mod\ \pi)$$
 with $\left| {{c_0}} \right| + \left| {{c_1}}
\right| + \left| {{c_2}} \right| + \left| {{c_3}}
 \right|+\left| {{c_4}} \right|$ minimum.
\end{definition} Also, we define Hurwitz weight of $\gamma =
\alpha -\beta$ as $$w_H(\gamma)=d_H(\alpha,\beta).$$
 It is possible to show that $d_H(\alpha, \beta)$ is a
 metric. We only show that the triangle inequality holds since the other conditions are straightforward. For this, let $\alpha$,
 $\beta$, and $\gamma$ be any three elements of $\mathcal{H}_\pi$. We
 have

 i) $d_H(\alpha, \beta)=w_H(\delta_1)=\left| {{a_0}} \right| + \left| {{a_1}}
\right| + \left| {{a_2}} \right| + \left| {{a_3}}
 \right|+ \left| {{a_4}} \right|  $, where $\delta_1 \equiv \alpha - \beta ={a_0} + {a_1}\widehat{e}_1 + {a_2}\widehat{e}_2 +
 {a_3}\widehat{e}_3+ {a_4}w \ (mod\  \pi )$ is an element of $\mathcal{H}_\pi$, and
 $\left| {{a_0}} \right| + \left| {{a_1}}
\right| + \left| {{a_2}} \right| + \left| {{a_3}}
 \right|+ \left| {{a_4}} \right|  $ is minimum.\

\

ii) $d_H(\alpha, \gamma)=w_H(\delta_2)=\left| {{b_0}} \right| +
\left| {{b_1}} \right| + \left| {{b_2}} \right| + \left| {{b_3}}
\right|+ \left| {{b_4}} \right| $, where $\delta_2 \equiv \alpha -
\gamma ={b_0} + {b_1}\widehat{e}_1 + {b_2}\widehat{e}_2 +
 {b_3}\widehat{e}_3 + {b_4}w \ (mod\  \pi )$ is an element of $\mathcal{H}_\pi$, and
 $\left| {{b_0}} \right| + \left| {{b_1}}
\right| + \left| {{b_2}} \right| + \left| {{b_3}}
 \right|+ \left| {{b_4}} \right|  $ is minimum.\

\

 iii) $d_H(\gamma, \beta)=w_H(\delta_3)=\left| {{c_0}} \right| +
\left| {{c_1}} \right| + \left| {{c_2}} \right| + \left| {{c_3}}
 \right| + \left| {{c_4}} \right|$, where $\delta_3 \equiv \gamma - \beta ={c_0} + {c_1}\widehat{e}_1 + {c_2}\widehat{e}_2 +
 {c_3}\widehat{e}_3+{c_4}w \ (mod\  \pi )$ is an element of $\mathcal{H}_\pi$, and
 $\left| {{c_0}} \right| +
\left| {{c_1}} \right| + \left| {{c_2}} \right| + \left| {{c_3}}
 \right| + \left| {{c_4}} \right|$ is minimum.\

 Thus, $\alpha-\beta=\delta_2+\delta_3 \ (\bmod \ \pi)$. However, $w_H\left( {{\delta _2} + {\delta _3}} \right) \ge w_H\left( {{\delta _1}}
 \right)$ since $w_H(\delta_1)=\left| {{a_0}} \right| + \left| {{a_1}}
\right| + \left| {{a_2}} \right| + \left| {{a_3}}
 \right|+ \left| {{a_4}}
 \right|$ is minimum. Therefore, $$d_H(\alpha, \beta) \le d_H(\alpha, \gamma)+d_H(\gamma, \beta).
 $$

Note that Hurwitz metric is not Lipschitz metric. To see this,
Lipschitz weight of the element
$w=\frac{1}{2}+\frac{1}{2}\widehat{e}_1+\frac{1}{2}\widehat{e}_2+\frac{1}{2}\widehat{e}_3$
is $w_L(w)=2$ and Hurwitz weight of the same element is
$w_H(w)=1$. \ \

The rest of this paper is organized as follows. In Section 2, one
error, double error and errors of arbitrary Hurwitz weight
correcting codes over $\mathcal{R}_\pi$ are defined. Also,
decoding algorithms of these codes are given. In Section 3, one
error, double error and errors of arbitrary Hurwitz weight
correcting codes over $\mathcal{H}_\pi$ are defined. Also,
decoding algorithms of these codes are given.

\section{Codes over $\mathcal{R}_\pi$}

Let $\pi$ be a prime in $\mathcal{R}$ and let $\beta$ be an
element of $\mathcal{R}_\pi$ such that $ \beta ^{(p-1)/6}  = \pm
w$. Recall that the cardinal number of $\mathcal{R}_\pi$ is equal
to $N(\pi)$. Thereafter, the length $n$ is taken as $n=(p-1)/6$,
where $p=\pi\pi ^{*}\equiv 1 \ (mod \ 6)$ is a prime in
$\mathcal{Z}$.

\begin{theorem}
Let $C$ be the code defined by the parity check matrix
\begin{equation}H = \left( {\begin{array}{*{20}{c}}
   {{1}}, & {{\beta }}, &  \cdots , & {\beta ^{n-1}}  \\
\end{array}} \right).
\end{equation}

Then $C$ can correct error vectors of Hurwitz weight 1 and some of
error vectors of Hurwitz weight 2. Error vectors of Hurwitz weight
1 have just one nonzero component. The nonzero component of the
above stated error vectors can take on one of the four values
$\pm1,\ \pm w$. The error vectors of Hurwitz weight 2 which can be
corrected have just one nonzero component which can take one of
the two values $\pm w^2$.

In other words, the code $C$ can correct any error pattern of the
form $e(x)=e_ix^{i}$, where $w_H(e_i)=1$ and the error patterns
$e(x)=\pm w^2 x^{i}$, where $w_H(\pm w^2)=2$. Thus, $d_H(C)\ge 3$.

\end{theorem}

\begin{proof} Let $r(x) = c(x)+e(x)$ be the received polynomial, where $c(x)$ denotes the codeword polynomial and $e(x)$ denotes the error polynomial.
 The vector corresponding to the polynomial $r(x)$ is $r=c+e$. We first compute the syndrome $S$ of $r$:

$$S=Hr^T=\beta^L.$$
By reducing $L$ modulo $n$, we determine the location of the error
with the value of the error $\beta^{L-l}$, where $l\equiv L \bmod
n$. Hence we have the location and the value of the error.
\end{proof}

\begin{example} Let $\pi  = 1 + 2\widehat{e}_1 + 2\widehat{e}_2 + 2\widehat{e}_3$ and $\beta =\widehat{e}_1+\widehat{e}_2+\widehat{e}_3 $.
Let $C$ be the code defined by the parity check matrix

$$H = \left[ {\begin{array}{*{20}{c}}
   {1,} & \beta   \\
\end{array}} \right].$$ Suppose that the received vector is $r=(-\beta ,\
w)$. The syndrome $S$ of $r$ is
$$S=Hr^T=-\frac{1}{2}(3+\widehat{e}_1+\widehat{e}_2+\widehat{e}_3)\equiv  \beta^5 \ (\bmod \ \pi).$$ The location of the error is $1\equiv 5 \ (\bmod \ 2)$ with the
value $\frac{\beta^5}{\beta}=\beta^4\equiv w^2 \ (\bmod \ \pi)$.
Hence, the corrected vector is $c=r-(0,\ w^2)=(-\beta, \
w-w^2)\equiv (-\beta, \ 1) \ (\bmod \ \pi)$.

\end{example}

\begin{theorem} Let $C$ be the code defined by the parity check matrix

$$H = \left[ {\begin{array}{*{20}{c}}
   {1,} & {\beta ,} & {{\beta ^2},} &  \cdots ,  & {{\beta ^4}}  \\
   {1,} & {{\beta ^7},} & {{\beta ^{14}},} &  \cdots , & {{\beta ^{7(n - 1)}}}  \\
\end{array}} \right].$$
Then $C$ is capable of correcting any error pattern of the form
$e(x)=e_ixî$, where $1\le w_H(e_i)\le d_{max}$. Here,
$d_{max}=max\left\{ {{w_H}\left( q \right):q \in \mathcal{R}_\pi}
\right\}$.
\end{theorem}

\begin{proof} Let $r = c+e$ be a received vector. First we compute the syndrome $S$ of $r$:

$$S = H{r^T} = \left( {\begin{array}{*{20}{c}}
   {{s_1} = {\beta ^{{L_1}}}}  \\
   {{s_7} = {\beta ^{7{L_1}}}}  \\
\end{array}} \right).$$

Let the error occurs in the location $l$, where
$\beta^{6l}=\frac{s_7}{s_1}$. By reducing $l\equiv L$ modulo $n$,
we determine the location of the error with the value of the error
$\frac{s_1}{\beta^{l}}$. Hence, we have the location and the value
of the error.

\end{proof}

\begin{example} Let $\pi  = 2 + 3\widehat{e}_1 + 3\widehat{e}_2 + 3\widehat{e}_3$ and $\beta =-\frac{1}{2}(5+\widehat{e}_1+\widehat{e}_2+\widehat{e}_3) $.
Let $C$ be the code defined by the parity check matrix

$$H = \left[ {\begin{array}{*{20}{c}}
   {1,} & {\beta ,} & {{\beta ^2},} &  \beta^3 , & {{\beta ^4}}  \\
   {1,} & {{\beta ^7},} & {{\beta ^{14}},} &  \beta^{21} , & {{\beta ^{28}}}  \\
\end{array}} \right].$$

 Suppose that the received vector is $r=(0 ,\
0,\ 0,\ 2,\ 0)$. The syndrome $S$ of $r$ is
$$S = H{r^T} = \left( {\begin{array}{*{20}{c}}
   {{s_1} = {\beta ^{{27}}}}  \\
   {{s_7} = {\beta ^{15}}}  \\
\end{array}} \right).$$
Then, $\beta^{6l}=\frac{s_7}{s_1}=\beta^{18}$ which implies that
$l=3 \ (\bmod \ n)$. Hence, the location of the error is $l=3$
with the value $\frac{s_1}{\beta^{l}}=\beta^{24}\equiv 2 \ (\bmod
\ \pi)$. Hence, the corrected vector is $c=r-(0 ,\ 0,\ 0,\ 2,\
0)=0$.

\end{example}

\begin{theorem} Let $C$ be the code defined by the parity check
matrix $$H = \left[ {\begin{array}{*{20}{c}}
   1, & \beta,  & {{\beta ^2}}, & {{\beta ^3}}, &  \cdots , & {{\beta ^{n - 1}}}  \\
   1, & {{\beta ^7}}, & {{\beta ^{14}}}, & {{\beta ^{21}}}, &  \cdots , & {{\beta ^{7(n - 1)}}}  \\
   1, & {{\beta ^{13}}}, & {{\beta ^{26}}}, & {{\beta ^{39}}}, &  \cdots , & {{\beta ^{13(n - 1)}}}  \\
\end{array}} \right].$$
Then $C$ can correct any error pattern of the form
$e(x)=e_ix^i+e_jx^j$, where $0\le w_H(e_i),\ w_H(e_j)\le 1$, and
$0\le i< j\le n-1$.

\end{theorem}

\begin{proof} If error vectors of Hurwitz weight $\le 2$ have only one nonzero
component exists, then the error can correct from Theorem 3. So,
suppose that double error occurs at two different components $l_1,
l_2$ of the received vector $r=c+e$. Its syndrome is

$$S=\left(%
\begin{array}{c}
  s_1 \\
  s_7 \\
  s_{13} \\
\end{array}%
\right).$$

The polynomial $\sigma (z)$ , which is help us to find the errors
location and the value of the errors, is computed as follows.
\begin{equation}  \label{eq:5} \sigma (z) = (z - \beta ^{l_1 } )(z
- \beta ^{l_2 } ) = z^2  - (\beta ^{l_1 } + \beta ^{l_2 } )z +
\beta ^{l_1 } .\beta ^{l_2 }  = z^2  - (s_1 )z +
\varepsilon,\end{equation} where $ \varepsilon$ is determined from
the syndromes. From $ s_1  = \beta ^{l_1 }  + \beta ^{l_2 } ,s_7 =
\beta ^{7l_1 }  + \beta ^{7l_2 }, s_{13}  = \beta ^{13l_1 }  +
\beta ^{13l_2 }$,  and $ \varepsilon  = \beta ^{l_1 + l_2 }$ we
get
$$s_1^{13}-s_{13}=1079\varepsilon^6s_1-2093\varepsilon^5s_1^3+910\varepsilon^4s_1^5-65\varepsilon^2s_1^9+13\varepsilon
s_1^{11}+156\varepsilon^3s_7$$ and
 $$s_1^7-s_7=7s_1\varepsilon^3-14s_1^3\varepsilon^2+7s_1^5\varepsilon .
$$ We now consider the polynomials
$$f(x)=1079s_1x^6-2093s_1^3x^5+910s_1^5x^4+156s_7x^3-65s_1^9x^2+13s_1^{11}x-s_1^{13}+s_{13}$$
and
$$g(x)=7s_1x^3-14s_1^3x^2+7s_1^5x -s_1^7+s_7,$$
where $f(x),g(x)\in \mathcal{R}_\pi[x]$. We prove that $f(x)$ and
$g(x)$ have only one root in common, that is, the degree of the
greatest common divisor polynomial of the polynomials $f(x)$ and
$g(x)$ is 1. To see this, we apply the Euclidean algorithm to
$f(x)$ and $g(x)$. Then we have

$$\begin{array}{*{20}{c}}
   {49{s_1}f(x) = q_1(x)g(x) + {r_1}(x) = (7553{s_1}{x^2} + 455s_1^3{x^2}}  \\
   { - 273s_1^5x + 78s_1^7 + 13{s_7})g(x) + 29s_1^{14} - 65s_1^7{s_7} - 13s_7^2}  \\
   { + 273s_1^{10}{x^2} - 273s_1^3{s_7}{x^2} - 182s_1^{12}x + 182s_1^5{s_7}x + 49{s_1}{s_{13}},}  \\
\end{array}$$
where the polynomials $q_1(x)$ and $r_1(x)$ denote the quotient
polynomial and the remainder polynomial, respectively. The
remainder polynomial $r_1(x)$ can not be the zero polynomial since
$$r_1(x)=91s_1^3(s_1^7-s_7)(3x^2-2s_1^2x+s_1^4)-62s_1^{14}+26s_1^7s_7-13s_7^2+49s_1s_{13}$$

and $s_1 \ne 0$, $s_1^7 \ne s_7$. Therefore, $g(x)$ does not
divide $f(x)$. To find out whether $r_1(x)$ has two common roots
with $g(x)$, we perform a second division such that
$$\begin{array}{*{20}{c}}{117s_1^2(s_7-s_1^7)g(x)=(4s_1^2-3x)r_1(x)-(4s_1^{14}+104s_1^7s_7} \\+39s_7^2-147s_1s_{13})x+
s_1^2(s_1^{14}+26s_1^7s_7+169s_7^2-196s_1s_{13}).\end{array}$$

Here, the remainder polynomial $r_2(x)$ is equal to $ t_1x+ t_0$,
where
\begin{equation}  \label{eq:5}
\begin{array}{*{20}{c}}
   {{t_1} = -(4s_1^{14}+104s_1^7s_7+39s_7^2-147s_1s_{13})}  \\
   {{t_0} = s_1^2(s_1^{14}+26s_1^7s_7+169s_7^2-196s_1s_{13}).}  \\
\end{array}
\end{equation}

The degree of the greatest common divisor polynomial of the
polynomials $f(x)$ and $g(x)$ is 1 since the polynomials $f(x)$
and $g(x)$ have only one common root. The root is
$x=-\frac{t_0}{t_1}$. In conclusion, $gcd(f(x),g(x))=r_2(x)$.
Hence, the proof is completed.
\end{proof}

Note that the roots of the polynomial \begin{equation}z^2  - (s_1
)z -\frac{t_0}{t_1}\end{equation} leads us to find the locations
of the errors and their values.

\begin{example} Let $\pi  = 2 + 3\widehat{e}_1 + 3\widehat{e}_2 + 3\widehat{e}_3$ and $\beta  =
-\frac{1}{2}(5+\widehat{e}_1+\widehat{e}_2+\widehat{e}_3)$. Let
$C$ be the code defined by the parity check matrix

$$H = \left[ {\begin{array}{*{20}{c}}
   1, & \beta,  & {{\beta ^2}}, & {{\beta ^3}}, & {{\beta ^4}}  \\
   1, & {{\beta ^7}}, & {{\beta ^{14}}}, & {{\beta ^{21}}}, & {{\beta ^{28}}}  \\
   1, & {{\beta ^{13}}}, & {{\beta ^{26}}}, & {{\beta ^{39}}}, & {{\beta ^{52}}}  \\
\end{array}} \right].$$
Suppose that the received vector is $r=(0,\ 0, \ \beta^{15},\ 0, \
\beta^5)$, where $\beta^{15}=-1,\ \beta^5=w$. We now apply the
decoding procedure in Theorem 5 to find the transmitted codeword.
The syndrome $S$ of $r$ is

$$S = H{r^T} = \left( {\begin{array}{*{20}{c}}
   {{s_1}}  \\
   {{s_7}}  \\
   {{s_{13}}}  \\
\end{array}} \right) = \left( {\begin{array}{*{20}{c}}
   {{\beta ^{17}} + {\beta ^9}}  \\
   {{\beta ^{29}} + {\beta ^{33}}}  \\
   {{\beta ^{41}} + {\beta ^{57}}}  \\
\end{array}} \right) \equiv \left( {\begin{array}{*{20}{c}}
   {{\beta ^8}}  \\
   {{\beta ^7}}  \\
   {{\beta ^{20}}}  \\
\end{array}} \right) \bmod \pi.$$

One can verify that $s_1^7 \ne s_7$, and $s_1^{13}\ne s_{13}$,
which shows that two errors have occurred. Using the formula (3),
we obtain $t_0=\beta^9$ and $t_1=\beta^{28}$. The roots of the
polynomial $z^2-s_1z-\frac{t_0}{t_1}$ are $z_1=\beta^{17}$, and
$z_2=\beta^{9}$. Therefore, the locations of the errors are
$2\equiv 17 \ (\bmod \ 5)$ and $4\equiv 9 \ (\bmod \ 5)$. Thus,
one error has occurred in location $l_1=2$ with the value
$\frac{\beta^{17}}{\beta^2}=-1$, and another one in location
$l_2=4$ with the value $\frac{\beta^{9}}{\beta^4}=w$. Hence, the
transmitted codeword is $c=(0,\ 0,\ 0, \ 0,\ 0)$.
\end{example}

\begin{theorem}  Let $C$ be the code defined by the parity check
matrix

$$H = \left[ {\begin{array}{*{20}{c}}
   1, & \beta,  & {{\beta ^2}}, &  \cdots , & {{\beta ^{n - 1}}}  \\
   1, & {{\beta ^7}}, & {{\beta ^{14}}}, &  \cdots,  & {{\beta ^{7(n - 1)}}}  \\
   1, & {{\beta ^{13}}}, & {{\beta ^{26}}}, &  \cdots , & {{\beta ^{13(n - 1)}}}  \\
   1, & {{\beta ^{19}}}, & {{\beta ^{38}}} ,&  \cdots , & {{\beta ^{19(n - 1)}}}  \\
\end{array}} \right].$$
Then $C$ is capable of correcting any error pattern of the form
$e(x)=e_ix^i+e_jx^j$, where $0\le w_H(e_i),\ w_H(e_j)\le d_{max}$,
with $0\le i<j\le n-1$. \end{theorem}

\begin{proof} Suppose that double error occurs at two different components
$l_1, l_2$ of the received vector $r=c+e$. Its syndrome is

$$S=\left(%
\begin{array}{c}
  s_1 \\
  s_7 \\
  s_{13} \\
  s_{19} \\
\end{array}%
\right).$$

From $ s_1  = \beta ^{l_1 }  + \beta ^{l_2 } ,s_7  = \beta ^{7l_1
}  + \beta ^{7l_2 }, s_{13}  = \beta ^{13l_1 }  + \beta ^{13l_2
},s_{19}  = \beta ^{19l_1 }  + \beta ^{19l_2 }$,  and $
\varepsilon  = \beta ^{l_1 + l_2 }$ we get

\begin{equation} \begin{array}{c}
 {s_1}{s_{13}} - s_7^2 = \left( {{\beta ^{{l_1}}} + {\beta ^{{l_2}}}} \right)\left( {{\beta ^{13{l_1}}} + {\beta ^{13{l_2}}}} \right) - {\left( {{\beta ^{7{l_1}}} + {\beta ^{7{l_2}}}} \right)^2} \\
  = \varepsilon {X^2} - 4{\varepsilon ^7}, \\
 \end{array}
\end{equation}
\begin{equation}{s_1}{s_{19}} - {s_7}{s_{13}} = \varepsilon {X^3} - 4{\varepsilon
^7}X,
\end{equation}
\begin{equation}{s_7}{s_{19}} - s_{13}^2 = {\varepsilon ^6}\left( {{X^2} -
4{\varepsilon ^7}} \right),\end{equation} where
$X=\beta^{6l_1}+\beta^{6l_2}$. Substituting (5) in (6) and (5) in
(7), we obtain

$$\begin{array}{l}
 \frac{{{s_1}{s_{19}} - {s_7}{s_{13}}}}{{{s_1}{s_{13}} - s_7^2}} = X = {\beta ^{6{l_1}}} + {\beta ^{6{l_2}}}, \\
 \frac{{{s_7}{s_{19}} - s_{13}^2}}{{{s_1}{s_{13}} - s_7^2}} = {\varepsilon ^6} = {\beta ^{6{l_1}}}{\beta ^{6{l_2}}}, \\
 \end{array}$$
respectively. We now consider the equation
\begin{equation}z^2-Xz+\varepsilon^6=0.\end{equation} The roots of the equation (8) give the
errors locations and their values.

\end{proof}

\section{Codes over $\mathcal{H}_\pi$}

In this section, we generalize codes from $\mathcal{R}_\pi$ to
$\mathcal{H}_\pi$. Our aim is to obtain codes correcting errors
coming from not only $\mathcal{R}_\pi$ but also $\mathcal{H}_\pi$.
Recall that the cardinal number of $\mathcal{H}_\pi$ is equal to
$2N(\pi)^2-1$. Let $\pi$ be a prime in $\mathcal{R}$ and let
$\beta$ be an element of $\mathcal{R}_\pi$ such that $ \beta
^{(p-1)/6}  = \pm w$, where $p = \pi {\pi ^ {*} }$.

\begin{theorem}
Let $C$ be the code defined by the parity check matrix
\begin{equation}H = \left( {\begin{array}{*{20}{c}}
   {{1}}, & {{\beta }} ,&  \cdots , & {\beta ^{n-1}}  \\
\end{array}} \right).
\end{equation}

Then $C$ can correct any error patterns of the form $e(x)=(\mu_1
w^t \mu_2) x^{i}$, where $w_H(\mu_1 w^t \mu_2)=1,2$ or $3$ with
$\mu_1, \mu_2 \in \left\{ { \pm 1, \pm \widehat{e}_1, \pm
\widehat{e}_2, \pm \widehat{e}_3} \right\}$and $t=0,1,2$.

\end{theorem}

Note that two quaternions $q_1,q_2 \in H(\mathcal{Z})$ are
associate if there exist unit quaternions $\mu_1,\mu_2$, such that
$q_1=\mu_1q_2\mu_2$ \cite{Ramo}.

\begin{proof} Let $r = c+e$ be a received vector. First we compute the syndrome $S$ of $r$:

$$S=Hr^T=\mu_1\beta^L\mu_2.$$
By reducing $L$ modulo $n$, we determine the location of the error
with the value of the error $\mu_1\beta^{L-l}\mu_2$, where
$l\equiv L \ (\bmod \ n)$. Hence, we have the location and the
value of the error.
\end{proof}

Feature of these codes is that these codes can correct more errors
than the codes over the ring $\mathcal{R}_\pi$ since these codes
can correct errors coming from not only $\mathcal{R}_\pi ^n$ but
also $\mathcal{H}_\pi ^n$.

\begin{example} Let $\pi  = 1 + 2\widehat{e}_1 + 2\widehat{e}_2 + 2\widehat{e}_3$ and $\beta
=\widehat{e}_1+\widehat{e}_2+\widehat{e}_3 $. Let $C$ be the code
defined by the parity check matrix

$$H = \left[ {\begin{array}{*{20}{c}}
   {1,} & \beta   \\
\end{array}} \right].$$ Suppose that the received vector is $r=(-\beta
,\frac{1}{2}(1+\widehat{e}_1-\widehat{e}_2-\widehat{e}_3))$. The
syndrome $S$ of $r$ is

$$\begin{array}{r}
 S = H{r^T} =  - 3 + 3{\widehat{e}_1} - 3{\widehat{e}_2} - 3{\widehat{e}_3} = {\widehat{e}_3}\left( { - 3 + 3{\widehat{e}_1} + 3{\widehat{e}_2} + 3{\widehat{e}_3}} \right){\widehat{e}_2}{\rm{                   }} \\
  \equiv {\widehat{e}_3}( - \frac{3}{2} - \frac{1}{2}{\widehat{e}_1} - \frac{1}{2}{\widehat{e}_2} - \frac{1}{2}{\widehat{e}_3}){\widehat{e}_2}\;\left( {\bmod \;\pi } \right) \\
 \end{array}$$
Here,
$-\frac{3}{2}-\frac{1}{2}\widehat{e}_1-\frac{1}{2}\widehat{e}_2-\frac{1}{2}\widehat{e}_3
\equiv \beta^5 \ (\bmod \ \pi) $. Thus, the location of the error
is $1\equiv 5 \ (\bmod \ 2)$ with the value
$\widehat{e}_3(\beta^{5-1})\widehat{e}_2\equiv \widehat{e}_3
w^2\widehat{e}_2 \ (\bmod \ \pi)$. Hence, the corrected vector is
$c=r-(0,\ \widehat{e}_3w^2\widehat{e}_2)=(-\beta, \ 1)$.

\end{example}

\begin{theorem} Let $C$ be the code defined by the parity check matrix

$$H = \left[ {\begin{array}{*{20}{c}}
   {1,} & {\beta ,} & {{\beta ^2},} &  \cdots , & {{\beta ^4}}  \\
   {1,} & {{\beta ^7},} & {{\beta ^{14}},} &  \cdots , & {{\beta ^{7(n - 1)}}}  \\
\end{array}} \right].$$
Then $C$ can correct any error vectors of Hurwitz weight $\le
d_{max}$, where $d_{max}=max\left\{ {{w_H}\left( q \right):q=\mu_1
q_1 \ or \ q=q_2 \mu_2,\ q \in \mathcal{H}_\pi, \ q_1,q_2 \in
\mathcal{R}_\pi,\ \mu_1,\mu_2 \in \left\{ { \pm 1, \pm
\widehat{e}_1, \pm \widehat{e}_2, \pm \widehat{e}_3} \right\}}
\right\}$. Error vectors of Hurwitz weight $\le d_{max}$ can be
corrected have just one nonzero component.

\end{theorem}

The proof can be easily seen from the proof of Theorem 4.

\begin{theorem} Let $C$ be the code defined by the parity check
matrix $$H = \left[ {\begin{array}{*{20}{c}}
   1, & \beta,  & {{\beta ^2}}, & {{\beta ^3}}, &  \cdots , & {{\beta ^{n - 1}}}  \\
   1 & {{\beta ^7}} & {{\beta ^{14}}} & {{\beta ^{21}}}, &  \cdots , & {{\beta ^{7(n - 1)}}}  \\
   1 ,& {{\beta ^{13}}}, & {{\beta ^{26}}} ,& {{\beta ^{39}}}, &  \cdots , & {{\beta ^{13(n - 1)}}}  \\
\end{array}} \right].$$
Then $C$ can correct some error vectors. The errors exist two
different components. If the form of the first error is $\mu _1
w^{t_1}$ (or $w^{t_1} \mu _2$), then the form of the second error
is $\pm \mu_1 w^{t_1}$ (or $\pm w^{t_2} \mu_2$) where
$t_1,t_2=0,1,2$.

\end{theorem}

\begin{proof}
Suppose that double error occurs at two different components $l_1,
l_2$ of the received vector $r=c+e$. Its syndrome is

$$S=rH^T\ (or \ Hr^T )=\left(%
\begin{array}{c}
  s_1=\mu_1 s_1^{'} \ (or \  s_1=s_1^{'} \mu_2)\\
  \\
  s_7=\mu_1s_7^{'}\ (or \  s_7=s_7^{'} \mu_2) \\
  \\
  s_{13}=\mu_1s_{13}^{'}\ (or \  s_{13}=s_{13}^{'} \mu_2) \\
 \end{array}%
\right),$$ where $\mu_1,\mu_2 \in \left\{ { \pm 1, \pm
\widehat{e}_1, \pm \widehat{e}_2, \pm \widehat{e}_3} \right\}$ and
$ s_1^{'}, s_7^{'}, s_{13}^{'}$ are elements of $\mathcal{R}$.
Using $s_1^{'},s_7^{'},s_{13}^{'}$, from Theorem 4, we can
determine $\frac{t_0}{t_1}$.

Assume that the roots of the polynomial $\sigma(z)$ are
$z_1=\beta^{L_1}$ and $z_2=\beta^{L_2}$. Then, the locations of
the errors are $l_1\equiv L_1 \bmod n$ with value
$\mu_1(\beta^{L_1-l_1} \ \bmod \pi) $ and $l_2\equiv L_2 \bmod n$
with value $\mu_1(\beta^{L_2-l_1}\ \bmod \pi) $. Hence, the proof
is completed.

\end{proof}

\begin{example} Let $\pi  = 2 + 3\widehat{e}_1 + 3\widehat{e}_2 + 3\widehat{e}_3$ and $\beta  =
-\frac{1}{2}(5+\widehat{e}_1+\widehat{e}_2+\widehat{e}_3)$. Let
$C$ be the code defined by the parity check matrix

$$H = \left[ {\begin{array}{*{20}{c}}
   1, & \beta,  & {{\beta ^2}}, & {{\beta ^3}}, & {{\beta ^4}}  \\
   1, & {{\beta ^7}}, & {{\beta ^{14}}}, & {{\beta ^{21}}}, & {{\beta ^{28}}}  \\
   1, & {{\beta ^{13}}}, & {{\beta ^{26}}}, & {{\beta ^{39}}}, & {{\beta ^{52}}}  \\
\end{array}} \right].$$

Suppose that the received vector is $r=(0,\ 0, \
\widehat{e}_2\beta^{15},\ 0, \ \widehat{e}_2\beta^{10})$, where
$\beta^{15}=-1,\ \beta^{10}=w^2$. We now apply the decoding
procedure in Theorem 7 to find the transmitted codeword. The
syndrome $S$ of $r$ is

$$S = r{H^T} = \left( {\begin{array}{*{20}{c}}
   {{s_1=\widehat{e}_2s_1^{'}}}  \\
   \\
   {{s_7=\widehat{e}_2s_7^{'}}}  \\
   \\
   {{s_{13}=\widehat{e}_2s_{13}^{'}}}  \\
\end{array}} \right) = \left( {\begin{array}{*{20}{c}}
   \widehat{e}_2({{\beta ^{17}} + {\beta ^9}})  \\
   \\
   \widehat{e}_2({{\beta ^{29}} + {\beta ^{33}}})  \\
   \\
   \widehat{e}_2({{\beta ^{41}} + {\beta ^{57}}} ) \\
\end{array}} \right) \equiv \left( {\begin{array}{*{20}{c}}
  \widehat{e}_2   \\
   \\
   \widehat{e}_2{{\beta ^{14}}}  \\
   \\
  \widehat{e}_2 {{\beta ^{17}}}  \\
\end{array}} \right) \bmod \pi.$$

One can verify that $(s_1^{'})^7 \ne s_7^{'}$, and
$(s_1^{'})^{13}\ne s_{13}{'}$, which shows that two errors have
occurred. Using the formula (3), we obtain $t_0=\beta^{21}$ and
$t_1=\beta^{5}$. The roots of the polynomial
$z^2-s_1^{'}z-\frac{t_0}{t_1}$ are $z_1=\beta^{17}$, and
$z_2=\beta^{14}$. Therefore, the locations of the errors are
$2\equiv 17 \ (\bmod \ 5)$ and $4\equiv 9 \ (\bmod \ 5)$. Thus,
one error has occurred in location $l_1=2$ with the value
$\widehat{e}_2\frac{\beta^{17}}{\beta^2}=-\widehat{e}_2$, and
another one in location $l_2=4$ with the value
$\widehat{e}_2\frac{\beta^{14}}{\beta^4}=\widehat{e}_2w^2$. Hence,
the transmitted codeword is $c=(0,\ 0,\ 0, \ 0,\ 0)$.

\end{example}

\begin{theorem} Let $C$ be the code defined by the parity check
matrix

$$H = \left[ {\begin{array}{*{20}{c}}
   1, & \beta , & {{\beta ^2}}, &  \cdots , & {{\beta ^{n - 1}}}  \\
   1 ,& {{\beta ^7}} ,& {{\beta ^{14}}}, &  \cdots , & {{\beta ^{7(n - 1)}}}  \\
   1, & {{\beta ^{13}}}, & {{\beta ^{26}}} ,&  \cdots , & {{\beta ^{13(n - 1)}}}  \\
   1, & {{\beta ^{19}}} ,& {{\beta ^{38}}}, &  \cdots , & {{\beta ^{19(n - 1)}}}  \\
\end{array}} \right].$$
Then $C$ is capable of correcting any error pattern of the form
$e(x)=e_ix^i+e_jx^j$, where $0\le w_H(e_i),\ w_H(e_j)\le d_{max}$,
with $0\le i<j\le n-1$, where $d_{max}$ defined in Theorem 8.
\end{theorem}
The proof of Theorem 10 can be easily seen from the proof of
Theorem 6.

\section{Conclusions}
In this study, the codes over a specific finite field
$\mathcal{R}_\pi$ with respect to a new metric called Hurwitz
metric are defined and decoding algorithms of these codes are
given. Using codes over $\mathcal{R}_\pi$, the codes correcting
errors coming from $\mathcal{H}_\pi$ are obtained.

\end{document}